\documentclass[11pt,a4paper ]{article}
\usepackage{cite,graphicx}
\usepackage{listings}
\usepackage{amsmath, amsthm, amssymb}
\usepackage{verbatim}
\usepackage{algorithm}
\usepackage{wrapfig}
\usepackage{a4wide}
\usepackage{xspace}
\usepackage[usenames,dvipsnames]{color}

\newtheorem{theorem}{Theorem}

\newtheorem{definition}{Definition}
\newtheorem{lemma}{Lemma}

\definecolor{light-gray}{gray}{0.95}

\newcommand{\cas}{CAS\xspace}
\newcommand{\scas}{\texttt{scas}\xspace}

\begin{document}

\begin{titlepage}

\title{
  \raisebox{30mm}[0mm][0mm]{\Large
    Technical Report no. 2009-10
  }
  \raisebox{5mm}[0mm][0mm]{
    \textbf{
      \begin{tabular}{c}
      Supporting Lock-Free Composition\\
      of Concurrent Data Objects
      \end{tabular}
    }
  }
}
\author{\raisebox{-15mm}[0mm][0mm]{\textit{\Large Daniel Cederman\footnote{Supported by Microsoft Research through its European PhD Scholarship Programme.}}} \and \raisebox{-15mm}[0mm][0mm]{\textit{\Large Philippas Tsigas\footnote{Partially supported by the Swedish Research Council (VR).}}}}
\date{
  \vspace{\stretch{1}}
  \enlargethispage{1.1\baselineskip}
  \includegraphics{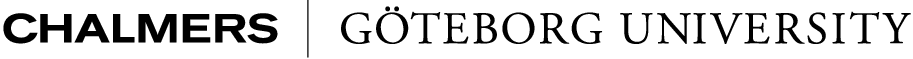} \\
  \vspace{5mm}
  \includegraphics{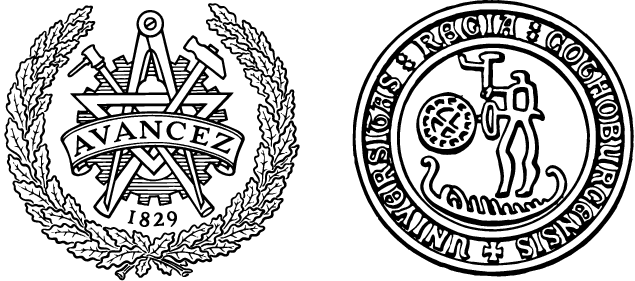} \\
  \vspace{12mm}
  Department of Computer Science and Engineering \\
  Chalmers University of Technology
    and G\"{o}teborg University \\
  SE-412 96 G\"{o}teborg, Sweden \\
  \vspace{12mm}
  G\"{o}teborg, 2009
}
\maketitle \thispagestyle{empty}
\end{titlepage}

\newpage
\thispagestyle{empty}
\mbox{}
\vspace{\stretch{1}}

\noindent
{\large
  \begin{tabular}{l}
    \includegraphics{graphs/logos/ChalmGUmarke} \\[3ex]
    Technical Report in Computer Science and Engineering at \\
    Chalmers University of Technology and G\"oteborg University
    \vspace{3ex} \\
    Technical Report no. 2009-10 \\
    ISSN: 1650-3023
    \vspace{3ex} \\
    Department of Computer Science and Engineering \\
    Chalmers University of Technology
      and G\"{o}teborg University \\
    SE-412 96 G\"{o}teborg, Sweden
    \vspace{3ex} \\
    G\"{o}teborg, Sweden, 2009
  \end{tabular}
}

\newpage

\thispagestyle{empty}

\begin{abstract}
Lock-free data objects offer several advantages over their blocking counterparts, such as being
immune to deadlocks and convoying and, more importantly, being highly concurrent. But they
share a common disadvantage in that the operations they provide are difficult to compose
into larger atomic operations while still guaranteeing lock-freedom. We present a lock-free
methodology for composing highly concurrent linearizable objects together by unifying their linearization points.
This makes it possible to relatively easily introduce atomic lock-free move operations to a wide range
of concurrent objects. Experimental evaluation has shown that the operations originally supported by
the data objects keep their performance behavior under our methodology.
\end{abstract}
\bigskip

\centerline{{\bf Keywords}: Lock-free, Composition, Data Structure, Move Operation}

\newpage
\thispagestyle{empty}
\mbox{}
\newpage
\setcounter{page}{1}

\section{Introduction}
A concurrent data object is lock-free if it guarantees that at least one operation, in the set of concurrent operations that it supports,
finishes after a finite number of its own steps have been executed by processes accessing the concurrent data object.
Lock-free data objects offer several advantages over their blocking counterparts, such as
being immune to deadlocks, priority inversion, and convoying, and have been shown to work
well in practice. They have been included in Intel's Threading Building Blocks Framework and the Java concurrency package,
and will be included in the forthcoming parallel extensions to the Microsoft .NET Framework \cite{tbb,java,msconc}.
However, the lack of a general, efficient, lock-free method for composing them makes it difficult for
the programmer to perform multiple operations together atomically in a complex software setting.
Algorithmic designs of lock-free data objects only consider the basic operations that define the data object.
To glue together multiple objects, one usually needs to solve a task that is many times more challenging than
the design of the data objects themselves, as lock-free data objects are often too complicated to be trivially altered.
Composing blocking data objects also puts the programmer in a difficult situation, as it requires knowledge of the way locks are handled internally,
in order to avoid deadlocks.

Techniques such as Software Transactional Memories (STMs) provide good composability \cite{compstm}, but have
problems with high overhead and have poor support for dealing with non-transactional code \cite{stmtoy,whystm}. They
require, with few exceptions, that the data objects are rewritten to be handled completely inside the STM,
which lowers performance compared to pure non-blocking data objects, and, moreover, provides no support to non-transactional code.

\subsection{Composing}
With the term composing we refer to the task of binding together multiple operations in such a way that they can be performed as one, without any intermediate state being
visible to other processes. In the literature the term is also used for nesting, making one data object part of another, which is an interesting problem, but
outside the scope of this paper. To give an example of the type of composing we consider, one can imagine a scenario where one wants to
compose together a hash-map and a linked list to provide a move operation for the user, in addition to all the previous operations that define
the two data objects.
Both data objects will most likely have implementations of insert and remove operations,
but an operation that can move elements from one instance of the object to another is not typically considered in the design of a single object.
Even if one of the data objects would have such an operation, it is unlikely that it would be compatible with the other data object, as it is of a different type.
However, if the insert and remove operations could be composed, the resulting move operation would possess those characteristics.

Composing lock-free concurrent data objects, in the context that we consider in this paper, has been an open problem in the area of lock-free data objects.
Customized compositions of specific concurrent data objects include the composition of lock-free flat-sets by Gidenstam et al. \cite{nbmalloc}
that constitute the foundation of a lock-free memory allocator.

Using blocking locks to compose the operations would reduce the concurrency and remove the lock-freedom guarantees of the remove and insert operations,
as it is not possible to combine lock-free operations with lock-based ones in a lock-free manner. This is because both the remove and the insert operations would have
to acquire a lock before executing, in order to ensure that they are not executed concurrently with the composed move operation. This would cause the operations to
be executed sequentially and lose their lock-free behavior, which guarantees that a process never needs to wait for
another process that is not making progress, which is what happens when a process needs to wait for a lock to be released.
Simply put, a generic way to compose concurrent objects, without foiling the possible lock-freedom guarantees of the objects, has to be lock-free itself.

\subsection{Contributions}

\begin{figure*}[t]
\center
\includegraphics[width=0.9\textwidth]{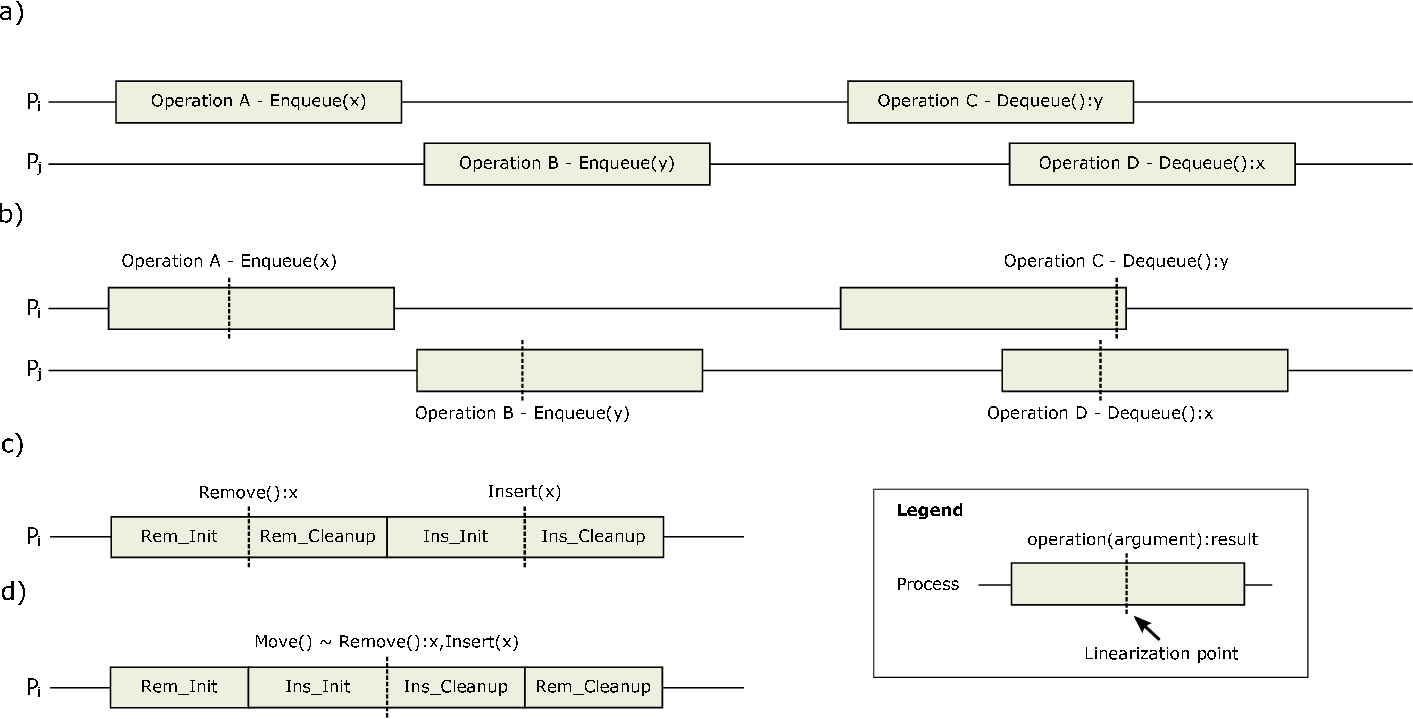}
\caption{(a) Concurrent history where two processes, $p_i$ and $p_j$, each enqueue and dequeue an element from a FIFO queue. (b) The same history with
the respective linearization points marked with dotted lines. (c) History of one process moving an element by means of a remove and insert operation. Notice
the time between the two linearization points. (d) The same history using the methodology presented in this paper. The linearization points have now been
unified.}
\label{fig:line}
\end{figure*}

The main contribution of this paper is to provide a methodology to introduce atomic move operations
that can move elements between objects of different types to a large class of already existing concurrent
objects without having to make significant changes to them. It manages this while preserving the lock-free
guarantees of the object and without introducing significant performance penalties to the previously supported
operations. Move operations are an important part of the core functionality needed when composing any
kind of containers, as they provide the possibility to shift items between objects.

We first present a set of properties that can be used to identify suitable concurrent objects and then we
describe the mostly mechanical changes needed for the move operation to function together with the object.

Our methodology is based on the idea of decomposing and then arranging lock-free operations appropriately
so that their linearization points can be combined to form new composed lock-free operations.
The linearization point of a concurrent operation is the point in time where the operation can be
said to have taken effect. Most concurrent data objects that are not read- or write-only support an insert and
a remove operation, or a set of equivalent operations that can be used to modify its content. These two types
of operations can be composed together using the method presented in this paper to make them appear to take
effect simultaneously. By doing this we provide a lock-free atomic operation that can move elements between objects
of different types. To the best of our knowledge this is the first time that such a general scheme has been proposed.

As a proof of concept we show how to apply our method on two commonly used
concurrent data objects, the lock-free queue by Michael and Scott \cite{msqueue} and the lock-free stack
by Treiber \cite{lfstack}. Experimental results on an Intel multiprocessor
system show that the methodology presented in the paper, applied to the previously mentioned lock-free
implementations, offers significantly better performance and scalability than a composition method based on locking. The proposed method
does this in addition to its qualitative advantages regarding progress guarantees that lock-freedom offers.
Moreover, the experimental evaluation has shown that the operations
originally supported by the data objects keep their performance behavior
while used as part of our methodology.

\section{The Model}
The model considered is the standard shared memory model, where a set of memory locations can be read from
and written to, by a number of processes that progress asynchronously.
Concurrent data objects are composed of a subset of these memory locations together with a set of operations that can
use read and write instructions, as well as other atomic instructions, such as compare-and-swap (\cas).
We require all concurrent data objects to be linearizable to assure correctness.

Linearizability is a commonly used correctness criterion introduced by Herlihy and Wing \cite{linear}. Each operation
on a concurrent object consists of an invocation and a response. A sequence of such
operations makes up a history. Operations in a concurrent history can be placed in any order if they occur concurrently,
but an operation that finishes before another one is invoked must appear before the latter.
If the operations in any actual concurrent history can be reordered in this way,
so that the history is equivalent to a correct sequential history, then the concurrent object is linearizable.

In Figure \ref{fig:line}a we see four operations being performed on a FIFO queue. Operation A has given
a response before operation B was invoked, which means that operation A must appear before operation
B in a sequential history for the object to be linearizable. Operations C and D, however, occur concurrently
and even though it seems like the operations violate the FIFO order there is no way to tell where operations
C and D actually take effect, which means that D occurring before C is a valid order and the object
is linearizable.

A way of looking at linearizability is to think that an operation takes effect at a specific point
in time, the linearization point. All operations can then be ordered according to the linearization
point to form a sequential history. In Figure \ref{fig:line}b we see the linearization points clearly
marked with dotted lines, which makes it easy to see that the equivalent sequential history is [A,B,D,C], which is a correct sequential history.

In Figure \ref{fig:line}c we see an invocation of a remove and insert operation that moves an element
from one object to another. It can easily be seen that there is a moment between the first linearization point,
in the remove operation, and the second linearization point, in the insert operation, where
the element being moved is not present in either object. This state could be seen by a concurrent process, which
is not always desirable. In this case it would be useful to be able to compose the two operations
so that they are performed atomically without other processes being able to see an intermediate state.

\section{The Methodology}
We present a method that can be used to unify the linearization points of a remove and an insert
operation for any two concurrent objects, given that they fulfill certain requirements. We call a concurrent
object that fulfills these requirements a \emph{move-candidate} object.

\subsection{Characterization}

\begin{definition}
\label{def:move-cand}
A concurrent object $S$ is a \emph{move-candidate} if it fulfills the following requirements:
\begin{enumerate}
\item It implements linearizable operations for insertion and removal of a single element.  \label{req:twoops}
\item Insert and remove operations invoked on different instances of the object can succeed simultaneously.\label{req:conc}
\item The linearization points of the successful insert and remove operations can be associated to successful \cas operations, on
a pointer, by the process that invoked it. Such an associated successful \cas can never lead to an unsuccessful insert or remove operation. \label{req:caslin}
\item The element to be removed is accessible before the linearization point.
\label{req:nogap}
\end{enumerate}
\end{definition}

To implement a move operation the equivalent of a remove and insert operation needs to be available or be implemented.
A generic insert or remove operation would be very difficult to write, as it must be tailored specifically to
the concurrent object, which motivates the first requirement.

Requirement 2 is needed since a move operation tries to perform the removal and insertion of an
element at the same time. If a successful removal invalidates an insertion, or the other way around,
then the move operation can never succeed. This could happen when the insert and remove operations share
locks between them or when they are using memory management schemes such as hazard pointers \cite{hazard}.
With shared locks there is the risk of deadlocks, when the process could be waiting for itself to release the lock in the remove
operation, before it can acquire the same lock in the insert operation. Hazard pointers, which are used to mark memory
that cannot yet be reused, could be overwritten if the same pointers are used in both the insert and remove operations.

Requirement 3 requires that the linearization points can be associated to successful \cas operations.
The linearization points are usually provided together with the algorithmic description of each object.
Implementations that use the LL/SC\footnote{LL (Load-Link) and SC (Store-Conditional) are used together. LL reads a value from a memory
location and SC can then only write a new value at the same location if the memory location has not been written to since the last LL.}
pair  for synchronization can be translated to ones that use CAS by
using the construction by Doherty et al. that implements the LL/SC functionality from \cas \cite{castollsc}.
The requirement also states that the \cas operation should be on a variable holding
a pointer. This is not a strict requirement; the reason for it is that the DCAS operation used in our methodology often needs to
be implemented in software due to lack of hardware support for such an operation. By only working with pointers it makes the implementation
much simpler. The last part, which requires the linearization point of an operation to be part of the process that invoked it, prevents concurrent
data objects from using \emph{helping} schemes.

Requirement 4 is necessary since the insert operation needs to be invoked with the removed element as an argument. The
element is usually available before the linearization point, but there are data objects where the element is never returned
by the remove operation, or is accessed after the linearization point for efficiency reasons.

\subsection{The Algorithm}
The main part of the algorithm is the actual move operation, which is described in the following section. Our move operation makes heavy
use of a DCAS operation that is described in detail in Section \ref{sec:dcas}.

\subsubsection{The Move Operation}

The main idea behind the move operation is based on the observation that the linearization points of many
concurrent objects' operations is a \cas and that by combining these {\cas}s and performing them simultaneously,
it would be possible to compose operations.
A move operation does not need an expensive general multi-word \cas, so an efficient two word \cas customized for this particular operation is good enough.
We would like to simplify the utilization of this idea
as much as possible, and for this reason we worked towards three goals when we designed the move
operation:

\begin{itemize}
\item Changes required to adapt the code that implements the operations of the concurrent data object should be minimal and possible to perform mechanically.
\item It should minimize the performance impact on the normal operations of the concurrent data objects.
\item The move operation should be lock-free if the insert and remove operations are lock-free.
\end{itemize}

With these goals in mind we decided that the easiest and most generic way would be to reuse the remove and insert operations that
are already supported by the object.
By definition a move-candidate operation has a linearization point that consists of a successful \cas. We
call the part of the operation prior to this linearization point the init-phase and the part after it the cleanup-phase.
The move can then be seen as taking place in five steps:

\begin{description}
\item[1st step.] The init-phase of the remove is performed. If the removal fails, due for example to the element not existing, the move is aborted.
Otherwise the arguments to the \cas at the potential linearization point
 are stored. By requirement 4 of the definition of a move-candidate we now have access to the element that is to be moved.
\item[2nd step.] The init-phase of the insert is performed using the element received in the previous step.
 If the insertion fails, due for example to the object being full, the move is aborted.
 Otherwise the arguments to the \cas at the potential linearization point are stored.
\item[3rd step.] The {\cas}s that define the linearization points, one for each of the two operations, are performed together atomically using a DCAS operation with the stored
\cas arguments. Step two is redone if DCAS failed due to a conflict in the insert operation. Steps one and two are redone if the conflict was in
the remove operation.
\item[4th step.] The cleanup-phase for the insert operation is performed.
\item[5th step.] The cleanup-phase for the remove operation is performed.
\end{description}

The above five steps of the algorithm are graphically described in Figure \ref{fig:line}d.

To be able to divide the insert and remove operations into the init- and cleanup-phases without resorting
to code duplication, it is required to replace all possible linearization point {\cas}s with a call to the \scas operation.
The task of the \scas operation is to restore control to the move operation and store the arguments intended for the \cas that was replaced.
The \scas operation is described in Algorithm \ref{alg:moveop} and comes in two forms, one to be called by the insert operations and one to
be called by the remove operations. They can be distinguished by the fact that the \scas for removal requires the element
to be moved as an argument. If the \scas operation is invoked as part of a normal insert or remove, it reverts back
to the functionality of a normal \cas. This should minimize the impact on the normal operations.

If the DCAS operation used is a software implementation that uses helping, it might be required to use hazard pointers to
disallow reclaiming of the memory used by it. In those cases the hazard pointers can be given as an argument to the \scas operation
and they will be brought to the DCAS operation. The DCAS operation provided in this paper uses helping and takes advantage of the
support for hazard pointers.

If the DCAS in step 3 should fail, this could be due to one of two reasons. First, it could fail because the \cas for the insert failed. In this
case the init-phase for the insert needs to be redone before the DCAS can be invoked again. Second, it could fail because
the \cas for the remove failed. Now we need to redo the init-phase for the remove, which means that the insert
operation needs to be aborted. For concurrent objects such as linked lists and stacks there might not be a way
for the insert to abort, so code to handle this scenario must be inserted.
If the insertion operation can fail for reasons other than conflicts with another operation, there is also a need for the
remove operation to be able to handle the possibility of aborting.

Depending on whether one uses a hardware implementation of a DCAS or a software implementation, it might also be required
to alter all accesses to memory words that could take part in DCAS, so that they access the word via a special read-operation
designed for the DCAS.

A concurrent object that is a move-candidate (Definition \ref{def:move-cand}) and has implemented all the above changes is called a \emph{move-ready}
concurrent object. This is described formally in the following definition.

\begin{definition}
A concurrent object is \emph{move-ready} if it is a \emph{move-candidate} and has
implemented the following changes:
\begin{enumerate}
\item The \cas at each linearization point in the insert and remove operations have been changed to \scas.
\item The insert (and remove) operation(s) can abort if the \scas returns \texttt{ABORT}.
\item All memory locations that could be part of a \scas are accessed via the \emph{read} operation.
\end{enumerate}
\end{definition}
The changes required are mostly mechanical once the object has been found to adhere to the move-ready definition.
This object can then be used by our move operation to move items between different instances
of any concurrent move-ready objects.

Theorem \ref{theorem:moveworks} in Section \ref{sec:proof} states that the move operation is linearizable and lock-free if used together
with two move-ready lock-free concurrent data objects.

\subsubsection{DCAS}
\label{sec:dcas}
The DCAS operation performs a \cas on two distinct words atomically (See Algorithm \ref{alg:dcassem} for its semantics).
It is unfortunately not commonly available in hardware, some say for good reasons \cite{dcassilver}, so for our experiments
it had to be implemented in software. There are several different multi-word compare-and-swap methods available in the literature
\cite{israelidisjoint,andersonwf,concobjects,reactivemw,universal,trans,stm,harrismwcas} and ours uses the same basic idea as in
the solution by Harris et al.

Lock-freedom is achieved by using a two-phase locking scheme
with helping\footnote{Lock-freedom does not exclude the use of locks, in contrast to its definition-name, if the locks can be revoked.}.
First an attempt is made to change both the words involved, using a normal \cas, to point to a descriptor that holds all
information required for another process to help the DCAS complete. See lines \ref{code:dcaslin1} and \ref{code:dcasword2set} in Algorithm \ref{alg:dcascode}.
If any of the {\cas}s fail, the DCAS is unsuccessful as both words need to match their old value.
In this case, if one of the {\cas}s succeeded, its corresponding word must be reverted back to its old value. When a word holds the descriptor
it cannot be changed by any other non-helping process, so
if both {\cas}s are successful, the DCAS as a whole is successful. The two words can now be changed one at a time to hold their respective new values.
See lines \ref{code:dcasnewval1} and \ref{code:dcasnewval2}.

If another process wants to access a word that is involved in a DCAS, it first needs to help the DCAS operation finish. The process knows that a word is used in a DCAS
if it is pointing to a descriptor. This is checked at line \ref{code:isdesccheck} in the read operation. In our experiments we have marked the descriptor pointer by setting
its least significant bit to one. This is a method introduced by Harris et al. \cite{bitmark} and it is possible to use since we assume that the word
will contain a pointer and that pointers will be aligned to the word size of the system.
Using the information in
the descriptor it tries to perform the same steps as the initiator, but marks the pointer
to the descriptor it tries to swap in with its thread id. This is done to avoid the ABA-problem, which can occur since \cas cannot distinguish a word
that has been changed from A to B and then back to A again, from a word whose value has remained A. Unless taken care of in this manner, the ABA-problem could
cause the DCAS to succeed multiple times, one for each helping process.

Our DCAS differs from the one by Harris et al. in that i) it has support for reporting which, if any, of the operations has failed, ii) it does not need to allocate an
 \texttt{RDCSSDescriptor} as it only changes two words, iii) it has support for hazard pointers, and iv) it requires two fewer {\cas}s in the uncontended case.
These are, however, minor differences and for our methodology to function it is not required to use our specific implementation.
Performance gains and practicality reasons account for the introduction of the new DCAS.
The DCAS is linearizable and lock-free according to Theorem \ref{theorem:dcasworks}.

\def\ListLetter#1{\def\listletter{#1}}
\renewcommand*\thelstnumber{{\tiny\listletter\the\value{lstnumber}}}
\lstset{escapeinside={(*}{*)}}

\lstset{emph={DCASDesc,DCAS},emphstyle=\bfseries, emph={[2]UNDECIDED,FIRSTFAILED,SUCCESS,SECONDFAILED,ABORT}, emphstyle={[2]\texttt}}

\ListLetter{NOTSET}

\begin{algorithm}[t]
\begin{lstlisting}[mathescape=true]
struct DCASDesc
  word old$_1$, old$_2$, new$_1$, new$_2$
  word $\ast$ptr$_1$, $\ast$ptr$_2$
  [word $\ast$hp$_1$, $\ast$hp$_2$]
  word res
\end{lstlisting}

\begin{lstlisting}[mathescape=true,numbers=none]
dres DCAS(desc)
  if(desc.$\ast$ptr$_1$ $\neq$ desc.old$_1$)
    return FIRSTFAILED
  if(desc.$\ast$ptr$_2$ $\neq$ desc.old$_2$)
    return SECONDFAILED
  desc.$\ast$ptr$_1$ $\gets$ desc.new$_1$
  desc.$\ast$ptr$_2$ $\gets$ desc.new$_2$
  return SUCCESS
\end{lstlisting}
\caption{Semantics of the DCAS operation.}
\label{alg:dcassem}
\end{algorithm}

\lstset{emph={DCAS,abort,scas,remove,insert,move,enqueue,dequeue,push,pop,read},emphstyle=\bfseries}

\begin{algorithm}[th]

\ListLetter{R}
\begin{lstlisting}[mathescape=true]%,numbers=left,numbersep=-8pt]

bool remove([key],$\ast$item)
  ...
  while(unsuccessful)
    ...
    result $\gets$ scas(ptr, old, new, element, [hp])
    // Only needed when insert can fail
    [if(result=ABORT)]
      [abort]
      [return false]
    ...
  ...
\end{lstlisting}

\ListLetter{I}
\begin{lstlisting}[mathescape=true]%,numbers=left,numbersep=-8pt]

bool insert([key],item)
  ...
  while(unsuccessful)
    ...
    result $\gets$ scas(ptr, old, new,[hp])
    if(result=ABORT)
      abort
      return false
    ...
  ...
\end{lstlisting}
\caption{Basic operations.}
\end{algorithm}

\begin{algorithm}[t]

\ListLetter{M}
\begin{lstlisting}[mathescape=true]
(*{\bf thread local variables}*)
  desc, ltarget, lskey, ltkey, insfailed
\end{lstlisting}

\begin{lstlisting}[mathescape=true,numbers=left,numbersep=-13pt, name=move]
   bool move(source, target, [skey, tkey])
    desc $\gets$ new DCASDesc
    desc.res $\gets$ UNDECIDED
    [lskey $\gets$ skey, ltkey $\gets$ tkey]
    ltarget $\gets$ target
    result $\gets$ source.remove([lskey], tmp)
    desc $\gets$ 0
    return result
\end{lstlisting}

\begin{lstlisting}[mathescape=true,numbers=left,numbersep=-13pt, name=move]
   fbool scas(ptr, old, new, element, [hp])
    if(desc $\neq$ 0)
      desc.ptr$_1$ $\gets$ ptr
      desc.old$_1$ $\gets$ old
      desc.new$_1$ $\gets$ new
      [desc.hp$_1$ $\gets$ hp]
      insfailed $\gets$ true
      result $\gets$ ltarget.insert([ltkey], element)
      if(insfailed)
        return ABORT
      return result
    else
      return cas(ptr,old,new)
\end{lstlisting}

\begin{lstlisting}[mathescape=true,numbers=left,numbersep=-13pt, name=move]
   fbool scas(ptr, old, new, [hp])
    if(desc $\neq$ 0)
      desc.ptr$_2$ $\gets$ ptr
      desc.old$_2$ $\gets$ old
      desc.new$_2$ $\gets$ new
      [desc.hp$_2$ $\gets$ hp]
      result $\gets$ DCAS(desc, true) (*\label{code:maindcascall}*)
      if(result != SUCCESS)
        desc $\gets$ new DCASDesc(desc)
        desc.res $\gets$ UNDECIDED
      insfailed $\gets$ false (*\label{code:insfailedset}*)
      if(result = FIRSTFAILED)
        return ABORT
      if(result = SECONDFAILED)
        return false
      return true
    else
      return cas(ptr,old,new)
\end{lstlisting}
\caption{Move operation.}
\label{alg:moveop}
\end{algorithm}

\begin{algorithm}[t]
\ListLetter{D}
\begin{lstlisting}[mathescape=true,numbers=left,numbersep=-12.5pt, name = dcas]
   dres DCAS(desc, initiator)
    [if($\lnot$initiator)]
        [hp$_1$ $\gets$ desc.hp$_1$, hp$_2$ $\gets$ desc.hp$_2$] (*\label{code:dcassethp}*)
    if(desc.res = SUCCESS $\lor$ SECONDFAILED) (*\label{code:dcashelpselftest}*)
      if(desc is marked)
        cas(desc.ptr$_2$, desc.old$_2$, desc) (*\label{code:dcashelpself1}*)
      else
        cas(desc.ptr$_1$, desc.old$_1$, desc) (*\label{code:dcashelpself2}*)
      return desc.res (*\label{code:dcasreturn1}*)
    if(initiator$\land \lnot$cas(desc.ptr$_1$,desc,desc.old$_1$)) (*\label{code:dcaslin1}*)
      return FIRSTFAILED (*\label{code:dcasreturn2}*)

    mdesc $\gets$ mark(unmark(desc),threadID)
    p2set $\gets$ cas(desc.ptr$_2$, mdesc, desc.old$_2$) (*\label{code:dcasword2set}*)
    if($\lnot$p2set)
      if($\ast$desc.ptr$_2$.ptr $\neq$ desc)
        cas(desc.res, SECONDFAILED, UNDECIDED) (*\label{code:dcaslin2}*) (*\label{code:dcasreschange1}*)
        if(desc.res = SUCCESS)
          return desc.res (*\label{code:dcasreturn3}*)
        if(desc.res = SECONDFAILED)
          cas(desc.ptr$_1$, desc.old$_1$, desc)
          return desc.res (*\label{code:dcasreturn4}*)

    cas(desc.res, mdesc, UNDECIDED) (*\label{code:dcaslin3}*) (*\label{code:dcasreschange2}*)
    if(desc.res = SECONDFAILED) (*\label{code:dcasfinalrescheck}*)
      if(p2set) cas(desc.ptr$_2$, desc.old$_2$, mdesc)
      return desc.res (*\label{code:dcasreturn5}*)
    cas(desc.ptr$_1$, desc.new$_1$, desc) (*\label{code:dcasnewval1}*)
    cas(desc.ptr$_2$, desc.new$_2$, desc.res) (*\label{code:dcasnewval2}*)
    desc.res $\gets$ SUCCESS (*\label{code:dcasreschange3}*)
    return desc.res (*\label{code:dcasreturn6}*)
\end{lstlisting}

\begin{lstlisting}[mathescape=true,numbers=left,numbersep=-12.5pt, name = dcas]
   word read($\ast$ptr)
    result $\gets \ast$ptr
    while(result is DCASDesc) (*\label{code:isdesccheck}*)
      hp$_d$ $\gets$ result
      if(hp$_d$ = $\ast$ptr)
        DCAS(result,false)
      result $\gets \ast$ptr
    return result
\end{lstlisting}
\caption{Double word compare-and-swap.}
\label{alg:dcascode}
\end{algorithm}

\section{Proof}
\label{sec:proof}

The first part of the proof section proves that the DCAS operation is linearizable and lock-free and
the second proves that the move operation is linearizable and lock-free.

\subsection{DCAS}

\begin{lemma}The DCAS descriptor's \texttt{res} variable can only change from \texttt{UNDECIDED}
to \texttt{SECONDFAILED} or from \texttt{UNDECIDED} to a marked
descriptor and consequently to \texttt{SUCCESS}.\label{lemma:dcasreschanges}\end{lemma}
\begin{proof}
The \texttt{res} variable is set at lines \ref{code:dcasreschange1}, \ref{code:dcasreschange2}, and \ref{code:dcasreschange3}.
On lines \ref{code:dcasreschange1} and \ref{code:dcasreschange2} the change is made using \cas, which assures that the variable
can only change from \texttt{UNDECIDED} to \texttt{SECONDFAILED} or to a marked descriptor. Line \ref{code:dcasreschange3} writes
\texttt{SUCCESS} directly to \texttt{res}, but it can only be reached if \texttt{res} differs from \texttt{SECONDFAILED} at line \ref{code:dcasfinalrescheck}, which
means that it must hold a marked descriptor as set on line \ref{code:dcasreschange2} or already hold \texttt{SUCCESS}.
\end{proof}

\begin{lemma}The initiating and all helping processes will receive the same result value.\label{lemma:dcassameres}\end{lemma}
\begin{proof}
DCAS returns the result value at lines \ref{code:dcasreturn1}, \ref{code:dcasreturn2}, \ref{code:dcasreturn3}, \ref{code:dcasreturn4}, \ref{code:dcasreturn5},
and \ref{code:dcasreturn6}. Lines \ref{code:dcasreturn4} and \ref{code:dcasreturn5} are only executed if \texttt{res} is equal to \texttt{SECONDFAILED} and
we know by Lemma \ref{lemma:dcasreschanges} that the result value cannot change after that. Lines \ref{code:dcasreturn6} and \ref{code:dcasreturn3} can
only be executed when \texttt{res} is \texttt{SUCCESS} and by the same Lemma the value can not change. Line \ref{code:dcasreturn1} only returns when the
result value is either \texttt{SUCCESS} or \texttt{SECONDFAILED} and as stated before these value cannot change. Line \ref{code:dcasreturn2} returns
\texttt{FIRSTFAILED} when the initiator process fails to announce the DCAS, which means that no other process will help the operation to finish.
\end{proof}

\begin{lemma}Iff the result value of the DCAS is \texttt{SUCCESS}, then $\ast$ptr$_1$ has changed value from old$_1$ to the descriptor to new$_1$ and $\ast$ptr$_2$
has changed value from old$_2$ to a marked descriptor to new$_2$ once.\label{lemma:dcassucciscorrect}\end{lemma}
\begin{proof}
On line \ref{code:dcaslin1} $\ast$ptr$_1$ is set to the descriptor by the initiating process as otherwise the result value would be \texttt{FIRSTFAIL}. On line
\ref{code:dcasword2set}, $\ast$ptr$_2$ is set to a marked descriptor by any of the processes. By contradiction, if all processes failed to change the
value of $\ast$ptr$_2$ on line \ref{code:dcasword2set}, the result value would be set to \texttt{SECONDFAILED} on line \ref{code:dcaslin2}. On line \ref{code:dcaslin3}
the \texttt{res} variable is set to point to a marked descriptor. This change is a step on the path to the \texttt{SUCCESS} result value and thus must be taken. On line
\ref{code:dcasnewval1} $\ast$ptr$_1$ is changed to new$_1$ by one process. It can only succeed once as the descriptor is only written once by the initiating process.
This is in contrast to $\ast$ptr$_2$ which can hold a marked descriptor multiple times due to the ABA-problem at line \ref{code:dcasword2set}. When $\ast$ptr$_2$
is changed to new$_2$ it could be changed back to old$_2$ by a process outside of the DCAS. The \cas at line \ref{code:dcasword2set} has no way of detecting this.
This is the reason why we are using a marked descriptor that is stored in the \texttt{res} variable using \cas, as this will allow only one process to change the value of
$\ast$ptr$_2$ to new$_2$ on line \ref{code:dcasnewval2}. A process that manages to store its marked descriptor to $\ast$ptr$_2$, but was not the first to set the
\texttt{res} variable, will have to change it back to its old value.
\end{proof}

\begin{lemma}Iff the result value of the DCAS is \texttt{FIRSTFAILED} or \texttt{SECONDFAILED}, then $\ast$ptr$_1$ was not changed to new$_1$ in the DCAS and $\ast$ptr$_2$
was not changed to new$_2$ in the DCAS due to either $\ast$ptr$_1\neq$old$_1$ or $\ast$ptr$_2\neq$old$_2$.\label{lemma:dcasfailiscorrect}\end{lemma}
\begin{proof}
If the \cas at line \ref{code:dcaslin1} fails, nothing is written to $\ast$ptr$_1$ by any processes since the operation is not announced. The \cas at line \ref{code:dcaslin3}
must fail, since otherwise the result value would not be \texttt{SECONDFAILED}. This means the test at line \ref{code:dcasfinalrescheck} will succeed and the operation will return
before line \ref{code:dcasnewval2}, which is the only place that $\ast$ptr$_2$ can be changed to new$_2$.
\end{proof}

\begin{lemma}Iff the result value of the DCAS is \texttt{SUCCESS}, then $\ast$ptr$_1$ held a descriptor at the same time as $\ast$ptr$_2$ held a marked descriptor. \label{lemma:dcasheldmarked}\end{lemma}
\begin{proof}
Line \ref{code:dcasnewval1} can only be reached if the {\cas}s at lines \ref{code:dcaslin1} and \ref{code:dcasword2set} were successful. The values
of $\ast$ptr$_1$ and $\ast$ptr$_2$ are not changed back until lines \ref{code:dcasnewval1} and \ref{code:dcasnewval2}, so just before the first process
reaches line \ref{code:dcasnewval1} $\ast$ptr$_1$ holds a descriptor and $\ast$ptr$_2$ holds a marked descriptor.
\end{proof}

\begin{lemma}If the initiating process protects $\ast$ptr$_1$ and $\ast$ptr$_2$ with hazard pointers, they will not be written to by any helping process
unless that process also protects them with hazard pointers.\label{lemma:dcasworkswithhp}\end{lemma}
\begin{proof}
If the initiating process protects the words, they will not be unprotected until that process returns, at which point the final result value must have been set.
This means that if the test at \ref{code:dcashelpselftest} fails for a helping process, the words were protected when the process local hazard pointers were
set at line \ref{code:dcassethp}. If the test did not fail, then the words are not guaranteed to be protected. But in that case the word that is written to at line
\ref{code:dcashelpself1} or \ref{code:dcashelpself2} is the same word that was read in the \texttt{read} operation. That word must have been protected
earlier by the process calling it, as otherwise it could potentially read from invalid space. Thus the words are either protected by the hazard pointers set at line
\ref{code:dcassethp} or by hazard pointers set before calling the \texttt{read} operation.
\end{proof}

\begin{lemma}DCAS is lock-free.\label{lemma:dcasislockfree}\end{lemma}
\begin{proof}
The only loop in DCAS is part of the read operation that is repeated until the word read is no longer a DCAS descriptor.
The word can be assigned the same descriptor, with different process id, for a maximum number of $p-1$ times, where $p$ is the number of processes in
the system. This can happen when each helping process manages to write to $\ast$ptr$_2$ due to the ABA-problem
mentioned earlier. This can only happen once for each process per descriptor as it will not get past the test on line
\ref{code:dcashelpselftest} a second time.

So, a descriptor appearing on a word means that either a process has started a new DCAS operation
or that a process has made an erroneous helping attempt. Either way, one process must have made progress
for this to happen, which makes the DCAS lock-free.
\end{proof}

\begin{theorem}The DCAS is lock-free and linearizable with possible linearization points at \ref{code:dcaslin1}, \ref{code:dcaslin2}, and \ref{code:dcaslin3},
and follows the semantics as specified in Algorithm \ref{alg:dcassem}.
\label{theorem:dcasworks}
\end{theorem}
\begin{proof}
Lemma \ref{lemma:dcassameres} gives that all processes return the same result value. According to Lemmata \ref{lemma:dcassucciscorrect}
and \ref{lemma:dcasfailiscorrect}, the result value can be seen as deciding the outcome of the DCAS. The result value is set at \ref{code:dcaslin2}
and \ref{code:dcaslin3}, which become possible linearization points. It is also set at \ref{code:dcasreschange3}, but that comes as a consequence of
the \cas at line \ref{code:dcaslin3}. The final candidate for linearization point happens when the \cas at line \ref{code:dcaslin3} fails. This happens before
the operation is announced so we do not need to set the \texttt{res} variable.

Lemma \ref{lemma:dcassucciscorrect}  proves that when the DCAS is successful it has changed both $\ast$ptr$_1$
and $\ast$ptr$_2$ to an intermediate state from a state where they were equal to old$_1$ and old$_2$, respectively.
Lemma \ref{lemma:dcasheldmarked} proves that they were in this intermediate state at the same time before they got their new values, according to Lemma \ref{lemma:dcassucciscorrect} again. If the DCAS was unsuccessful then nothing is changed due to either $\ast$ptr$_1\neq$old$_1$ or $\ast$ptr$_2\neq$old$_2$.
This is in accordance with the semantics specified in Algorithm \ref{alg:dcassem}.

Lemma \ref{lemma:dcasislockfree} gives that DCAS is lock-free.
\end{proof}

\subsection{Move Operation}

\begin{theorem} The \emph{move} operation is linearizable and lock-free if applied to two lock-free move-ready concurrent objects.
\label{theorem:moveworks}\end{theorem}
\begin{proof}
We consider DCAS an atomic operation as shown by Theorem \ref{theorem:dcasworks}.
All writes, except the ones done by the DCAS operation, are process local and can as such be ignored.

The move operation starts with an invocation of the remove operation. If it fails, it means that there were no elements to remove from the object and
that the linearization point must lie somewhere in the remove operation, since requirement 1 of the definition of a move-candidate
states that the operations should be linearizable.
If the process reaches the \scas call, the insert operation is invoked with the element to be removed as an argument. If the insert fails
it means that it was not possible to insert the element. In this case \texttt{insfailed} was not set at line \ref{code:insfailedset} and \scas will abort the remove operation.
In this case, the linearization point is somewhere in the insert operation. In both these scenarios, the move operation as a whole fails.

If the process reached the second \scas call, the one in the insert operation, the DCAS operation is invoked.
If it is successful, then both the insert and remove operation must have succeeded according to requirement 3 of the definition of a move-candidate.
By requirement 1, they can only succeed once, which makes the DCAS the linearization point.
If the DCAS fails nothing is written to the shared memory and either the insert or both the remove and the insert operations are restarted.

Since the insert and remove operations are lock-free, the only reason for the DCAS to fail is that another process has made progress in their insertion or
removal of an element. This makes the move operation as a whole lock-free.\end{proof}

\section{Case Study}
To get a better understanding of how our methodology can be used in practice, we apply it to two commonly used
concurrent objects, the lock-free queue by Michael and Scott \cite{msqueue} and the lock-free stack
by Treiber \cite{lfstack}. The objects use hazard pointers for memory management and the selection
of them is motivated in the paper by Michael \cite{hazard}.

\begin{algorithm}[t]
\ListLetter{Q}
\begin{lstlisting}[mathescape=true,numbers=left,numbersep=-12.5pt, name = queue]
   bool enqueue(val)
    node $\gets$ new Node
    node.next $\gets$ 0
    node.val $\gets$ val
    while(true)
      ltail $\gets$ read(tail) (*\label{code:qget1}*)
      hp$_1$ $\gets$ ltail; if(hp$_1$ != read(tail) continue (*\label{code:qget2}*)
      lnext $\gets$ read(ltail.next) (*\label{code:qget3}*)
      hp$_2$ $\gets$ lnext
      if(ltail != read(tail)) continue (*\label{code:qget4}*)
      if(lnext != 0)
        cas(tail,ltail,lnext)
        continue
      res $\gets$ scas(ltail.next,0,node,hp$_1$) (*\label{code:enqlin}*)
      if(res = abort)
        free node  (*\label{code:lfqfreemem}*)
        return false
      if(res = true)
        cas(tail,ltail,node)
        return true
\end{lstlisting}

\begin{lstlisting}[mathescape=true,numbers=left,numbersep=-12.5pt, name = queue]
   bool dequeue($\ast$val)
    while(true)
      lhead $\gets$ read(head) (*\label{code:qget5}*)
      hp$_3$ $\gets$ lhead; if(hp$_3$ != read(head) continue (*\label{code:qget6}*)
      ltail $\gets$ read(tail) (*\label{code:qget7}*)
      lnext $\gets$ read(lhead.next) (*\label{code:qget8}*)
      hp$_4$ $\gets$ lnext
      if(lhead!=read(head)) continue (*\label{code:qget9}*)
      if(lnext=0) return false (*\label{code:qfalselin}*)
      if(lhead==ltail)
        cas(tail,ltail,lnext)
        continue
      $\ast$val $\gets$ lnext.val (*\label{code:valread}*)
      if(scas(head,lhead,lnext,val,hp$_3$) (*\label{code:deqlin}*)
        free lhead
        return true
\end{lstlisting}
\caption{Lock-free queue by Michael and Scott \cite{msqueue}.}
\end{algorithm}

\subsection{Queue}
The first task is to see if the queue is a move-candidate as defined by Definition \ref{def:move-cand}:

\begin{enumerate}
\item The queue fulfills the first requirement by providing dequeue and enqueue operations, which have been shown
to be linearizable \cite{msqueue}.
\item The insert and remove operations share hazard pointers in the original implementation. By using a separate set
of hazard pointers for the dequeue operation we fulfill requirement number 2, as no other information is shared
between two instances of the object.
\item The linearization points can be found on lines \ref{code:deqlin}, and \ref{code:enqlin} and both consist of a
successful \cas, which is what requirement number 3 asks for. There is also a linearization point at line \ref{code:qfalselin},
but it is not taken in the case of a successful dequeue. These linearization points were provided together with the algorithmic description of the object, which is usually the case
for the concurrent linearizable objects that exist in the literature.
\item The linearization point for the dequeue is on line \ref{code:deqlin} and the value that is
read in case of a successful \cas is available on line \ref{code:valread},
which must be executed before line \ref{code:deqlin}.
\end{enumerate}

The above simple observations give us the following lemma in a straight forward way.
\begin{lemma}The lock-free queue by Michael and Scott is a move-candidate.\end{lemma}

After making sure that the queue is a move-candidate we need to replace the \cas operations at the linearization
points on lines \ref{code:deqlin} and \ref{code:enqlin}
with calls to the \scas operation. If we are using a software implementation of DCAS we also
need to alter all lines where words are read that could be part of a DCAS, so that they access them via the read
operation. For the queue these changes need to be done on lines \ref{code:qget1}, \ref{code:qget2}, \ref{code:qget3},
\ref{code:qget4}, \ref{code:qget5}, \ref{code:qget6}, \ref{code:qget7}, \ref{code:qget8}, and \ref{code:qget9}.

One must also handle the case of \scas returning \texttt{ABORT} in the enqueue.
 Since there has been no change to
the queue, the only thing to do before returning from the operation is to free up the allocated memory on line \ref{code:lfqfreemem}.
The enqueue cannot fail so there is no need to handle the \texttt{ABORT} result value in the dequeue operation.

The move operation can now be used with the queue. In Section \ref{sec:experiments} we evaluate the performance of
the move-ready queue  when combined with another queue, and when combined with the Treiber stack.

\begin{algorithm}[t]

\caption{Lock-free stack by Treiber \cite{lfstack}.}
\ListLetter{S}
\begin{lstlisting}[mathescape=true,numbers=left,numbersep=-12.5pt, name = stack]
   bool push(val)
    node $\gets$ new Node
    node.val $\gets$ val
    while(true)
      ltop $\gets$ read(top)  (*\label{code:topread1}*)
      node.next $\gets$ ltop
      res $\gets$ scas(top,ltop,node) (*\label{code:pushlin}*)
      if(res = abort)
        free node
        return false
      if(res = true)
        return true
\end{lstlisting}
\begin{lstlisting}[mathescape=true,numbers=left,numbersep=-12.5pt, name = stack]
   bool pop(val)
    while(true)
      ltop $\gets$ read(top)  (*\label{code:topread2}*)
      if(ltop = 0)
        return false (*\label{code:stackfalselin}*)
      hp $\gets$ ltop
      if(read(top) != ltop) (*\label{code:topread3}*)
        continue
      val $\gets$ ltop.val  (*\label{code:stackval}*)
      if(scas(top,ltop,ltop.next,val))  (*\label{code:poplin}*)
        free ltop
        return true
\end{lstlisting}


\end{algorithm}

\subsection{Stack}

Once again we first check to see if the stack fulfils the requirements of the move-candidate definition:

\begin{enumerate}
\item The push and pop operations are used to insert and remove elements and it has been shown that they are linearizable.
Vafeiadis has, for example, given a formal proof of this \cite{stackline}.
\item There is nothing shared between instances of the object, so the push and pop operations can succeed simultaneously.
\item The linearization points on lines \ref{code:pushlin} and \ref{code:poplin} are both \cas operations. The linearization point on line \ref{code:stackfalselin}
is not a \cas, but it is only taken when the source stack is empty and when the move can not succeed. The conditions in the definition only require
successful operations to be associated to a successful \cas.
\item The element to be removed is available on line \ref{code:stackval}, which is before the linearization point on line \ref{code:poplin}.
\end{enumerate}

The above simple observations give us the following lemma in a straight forward way.
\begin{lemma}The lock-free stack by Treiber is a move-candidate.\end{lemma}

To make the stack object move-ready we change the \cas operations on lines \ref{code:pushlin} and \ref{code:poplin} to point to \scas instead. We
also need to change the read of $top$ on lines \ref{code:topread1}, \ref{code:topread2}, and \ref{code:topread3}, if we are using a software implementation of DCAS,
so that it goes via the read operation.

Since push can be aborted we also need to add a check after line \ref{code:pushlin} that looks for this condition and frees allocated memory.

The stack is now move-ready and can be used to atomically move elements between instances of the stack and other move-ready objects, such as the previously described
queue. In Section \ref{sec:experiments} we evaluate the performance of
the move-ready stack when combined with another stack as well as when combined with the Michael and Scott queue.

\section{Experiments}
\label{sec:experiments}

\def\figsize{0.85}

\begin{figure*}[t]
\center
\includegraphics[width=\figsize\textwidth]{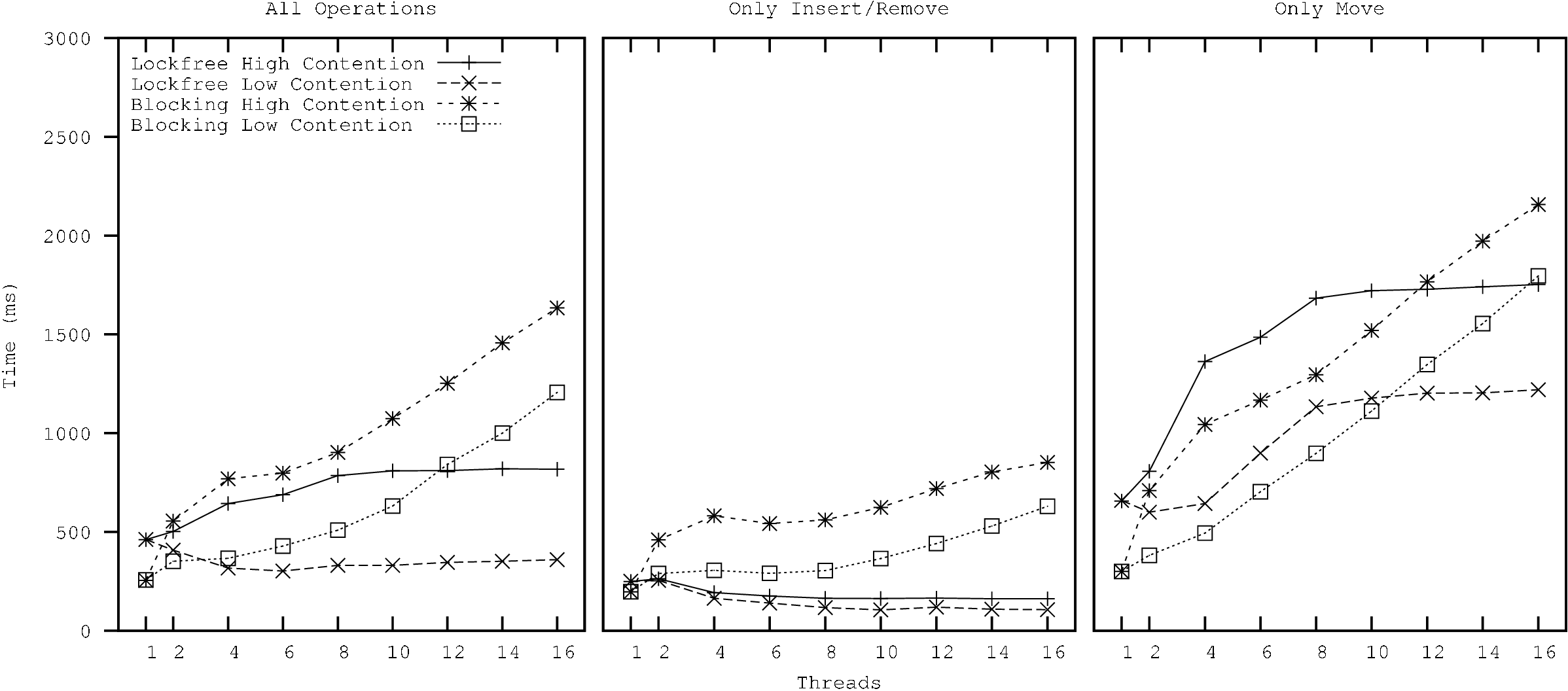}
\caption{Results from the queue/stack evaluation without back-off.}
\label{fig:qs}
\end{figure*}

\begin{figure*}[t]
\center
\includegraphics[width=\figsize\textwidth]{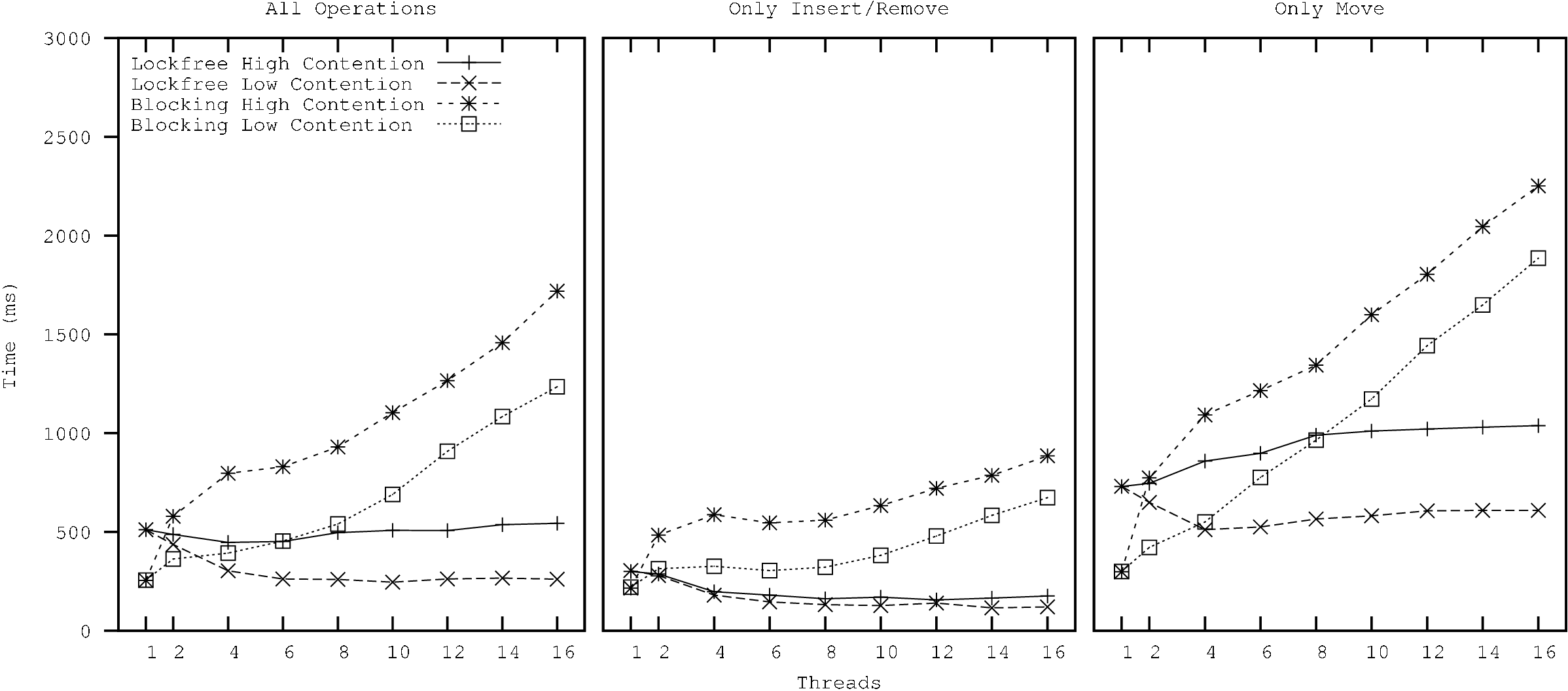}
\caption{Results from the queue evaluation without back-off.}
\label{fig:queue}
\end{figure*}

\begin{figure*}[t]
\center
\includegraphics[width=\figsize\textwidth]{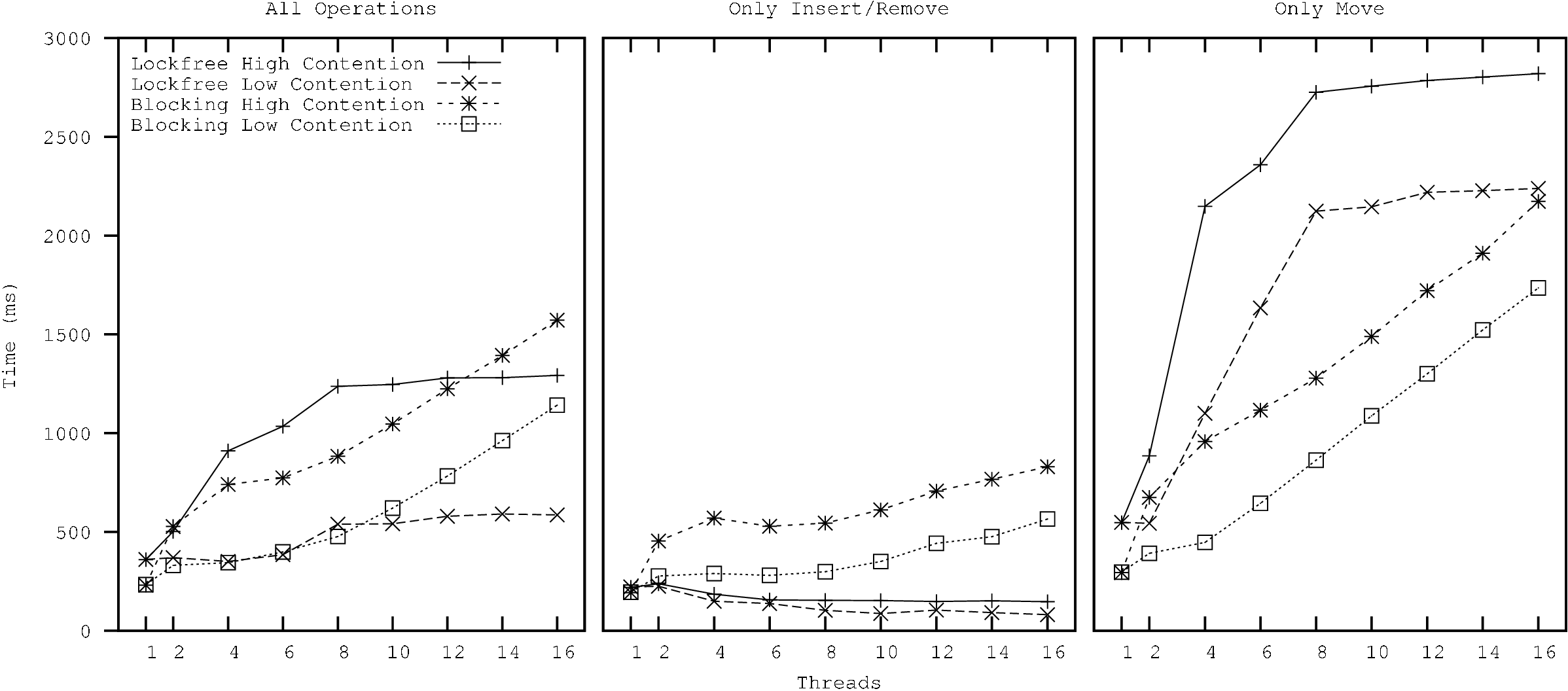}
\caption{Results from the stack evaluation without back-off.}
\label{fig:stack}
\end{figure*}

The evaluation was performed on a machine with an Intel Core i7 950 3\,GHz processor and 6\,GB DDR3-1333 memory.
All experiments were based on either two queues, two stacks, or one queue and one stack. Each thread
randomly performed operations from a set of either just move operations, or just insert/remove operations,
or both move and insert/remove operations. A total of five million operations were distributed evenly to
between one and sixteen threads and each trial was run fifty times.

For reference we compared the lock-free concurrent objects with simple blocking implementations using test-test-and-set
to implement a lock. We did the experiments both with and without a backoff function. The backoff function was used to lower
the contention so that every time a thread failed to acquire the lock or, in case of the lock-free objects, failed to insert or
remove an element due to a conflict, the time it waited before trying again was doubled. The starting wait time and the maximum
wait time were adjusted so as to give the best performance to the blocking implementation.

All implementations used the same lock-free memory manager. Freed nodes are placed on a local list
with a capacity of 200 nodes. When the list is full it is placed on a global lock-free stack. A process that
requires more nodes accesses the global stack to get a new list of free nodes. Hazard pointers were used
to prevent nodes in use from being reclaimed.

Two load distributions were tested, one with high contention and one with low contention, where each thread
did some local work for a variable amount of time after they had performed an operation on the object.
The work time is picked from a normal distribution and the work takes around $0.1\mu s$
per operation on average for the high contention distribution and $0.5\mu s$ per operation on the low contention
distribution.

The total time taken for all threads to finish their allotted operations with
no backoff function, excluding the time it took to perform the local work, is shown in Figures \ref{fig:qs}, \ref{fig:queue} and \ref{fig:stack}.
The local work time was subtracted from the result to emphasize the synchronization overhead.

\section{Discussion}
The results for only the remove/insert operations show that the lock-free versions scale with the number of threads,
while the blocking drops in performance when the contention rises. We get similar results when only move operations
are performed. The lock-free scales quite well, while the blocking performs worse as more contention is added in form of threads.
With backoff the result is similar, except for the blocking implementation that shows a better result for the high contention case.
However, it is typically hard to predict the contention level which often varies during runtime, making it difficult to design
an optimal backoff function that works well during both high and low contention.

The results for the lock-free stack when only move operations are performed are a bit poor. This can be partly attributed to the fact
that there is a lot of false helping in the DCAS, due to the ABA-problem that occurs when the same element is removed and then
inserted again. This causes a lot of extra {\cas}s to occur. The problem can be alleviated by adding a counter to the top pointer in the stack,
removing the possibility of the ABA-problem occurring. The downside with this solution is that it somewhat lowers the performance
of the normal insert and remove operations at the same time.

The graphs with the result of all operations can be seen as an average between the two other graphs and shows
that the lock-free data objects scale quite well. It should also be noted that it is not possible to combine a blocking move
operation with non-blocking insert/remove operations.

\section{Conclusion}
We present a lock-free methodology for composing
together highly concurrent linearizable objects by unifying their
linearization points.
Our methodology introduces atomic move operations that can move elements between
objects of different types, to a large class of already existing
concurrent objects without having to make significant changes to
them.

Our experimental results demonstrate that the methodology presented in the
paper, applied to the classical lock-free implementations,
offers significantly better performance and scalability than a
composition method based on locking. These results also demonstrate that it does not introduce significant performance penalties
to the previously supported operations of the concurrent objects.

Our methodology can also be easily extended to support $n$ operations on $n$ distinct objects, for example to create functions
that remove an item from one object and insert it into $n$ others atomically.

\bibliographystyle{abbrv}
\bibliography{reflist}

\end{document}